\documentclass[journal,twoside,web]{ieeecolor}

\usepackage{cite}
\usepackage{bm}
\usepackage{amsthm}
\usepackage{amsmath,amssymb,amsfonts}
\usepackage{algorithmic}
\usepackage{setspace}
\usepackage{graphicx}
\usepackage{textcomp}
\usepackage{xcolor} 
\def\BibTeX{{\rm B\kern-.05em{\sc i\kern-.025em b}\kern-.08em
    T\kern-.1667em\lower.7ex\hbox{E}\kern-.125emX}}
\usepackage{xcolor}
\definecolor{subsectioncolor}{RGB}{0, 51, 102}
\newtheorem{Thm}{Theorem}
\newtheorem{Lem}{Lemma}
\newtheorem{Prop}{Proposition}
\newtheorem{Def}{Definition}

\newtheorem{Rmk}{Remark}
\newtheorem{Exam}{Example}
\newtheorem{Prob}{\textbf{Problem}}
\newtheorem{Cor}{Corollary}

\title{\LARGE \bf
Tunable Input-to-State Safety with Input Constraints
}

\author{Ming Li$^{1}$, Jin Chen$^{2}$, Dimos V. Dimarogonas$^{1}$ 
\thanks{*This work was supported by the ERC Consolidator Grant LEAFHOUND, the Swedish Research Council (VR) and the Knut och Alice Wallenberg Foundation (KAW).}
\thanks{$^{1}$Ming Li and Dimos V. Dimarogonas are with the Division of Decision and Control Systems, EECS, KTH Royal Institute of Technology, 19255 Stockholm, Sweden
(e-mail: ming3@kth.se; dimos@kth.se).}
\thanks{$^{2}$Jin Chen is with the School of Mechanical Engineering, Shanghai Jiao Tong University, Shanghai, China (e-mail: chenjin920414@sjtu.edu.cn).%
}}

\begin{document}

\maketitle
\thispagestyle{empty}
\pagestyle{empty}

\begin{abstract}
Tunable input-to-state safety (TISSf) generalizes the input-to-state safety (ISSf) framework by incorporating a tuning function that regulates safety conservatism while preserving robustness against perturbations. Despite its flexibility, the TISSf tuning function is often designed without explicitly incorporating actuator limits, which can lead to incompatibility with input constraints. To address this gap, this paper proposes a framework that integrates general compact input constraints into tuning function synthesis. Leveraging a geometric perspective, we characterize the TISSf condition as a state-dependent half-space constraint and derive a verifiable certificate for input compatibility using support functions. This characterization transforms the compatibility requirement into a design constraint on the tuning function, yielding a prescriptive lower bound that defines an admissible family of tunings under input constraints. These results are specialized to norm-bounded, polyhedral, and box constraints, yielding tractable control design conditions. We show that these conditions, combined with tuning function monotonicity, guarantee input compatibility and recursive feasibility of the resulting quadratic program (QP)-based safety filter. Furthermore, an offline parameter selection procedure using a covering-based sampling strategy ensures compatibility across the entire safe set via a linear program (LP). A connected cruise control (CCC) application demonstrates robust safety under TISSf while enforcing input constraints by design.
\end{abstract}
\vspace{-6pt}
\section{Introduction}

In recent years, control barrier functions (CBFs) have attracted significant attention in safety-critical control due to their ability to handle nonlinear systems and manage high-relative-degree constraints in real-time~\cite{Zeroing_CBF,wieland2007constructive,xiao2021high,tan2021high}. However, the presence of perturbations (e.g., model uncertainties and disturbances) often undermines the safety guarantees of CBFs in real-world deployments if not appropriately addressed. To overcome this, researchers have extended the CBF framework by incorporating adaptive control~\cite{lopez2020robust,xiao2021adaptive}, learning-based methods~\cite{taylor2020learning,castaneda2021gaussian}, or robust control theory~\cite{jankovic2018robust,kolathaya2018input}.

Among these extensions, input-to-state safety (ISSf)~\cite{romdlony2016new,kolathaya2018input} has emerged as a powerful robustness notion because it quantifies safety degradation under bounded perturbations and integrates naturally with controller synthesis. Tunable input-to-state safety (TISSf)~\cite{alan2021safe}, which generalizes ISSf by introducing a tuning function to adjust the robustness margin, offers additional flexibility to regulate safety conservatism (e.g., to avoid overly cautious maneuvers far from obstacles). From a geometric perspective, this tuning mechanism shifts the hyperplane of the TISSf-CBF-induced half-space, providing a systematic way to regulate constraint conservatism.

However, a fundamental limitation of the TISSf framework in~\cite{alan2021safe} is that tuning functions are often designed solely to regulate safety conservatism, without explicitly enforcing compatibility with input constraints, which are critical in practical control design~\cite{agrawal2021safe,breeden2023robust}. For example, in~\cite{alan2021safe,alan2023control}, tuning parameters are selected via trial-and-error and validated only in specific scenarios, so input constraint satisfaction typically relies on case-by-case tuning or a posteriori saturation. This provides no theoretical guarantee of compatibility between the TISSf-CBF and input constraints. Therefore, it motivates a constructive framework for synthesizing tuning functions with compatibility guarantees, which ensures that TISSf-CBF and input constraints can be satisfied simultaneously by design.

This paper resolves this gap by reformulating the compatibility between TISSf and input constraints as a constructive tuning function design objective. Instead of treating compatibility as a property to be verified a posteriori~\cite{mestres2022optimization}, we provide a rigorous framework to determine the admissible class of tuning functions that inherently ensures compatibility. This approach provides systematic guidance for the synthesis of tuning functions, ensuring that TISSf and input constraints are simultaneously satisfied by construction. 

The main contributions of this paper are summarized as follows. We first establish a compatibility characterization for general compact input sets using support functions, which yields a verifiable lower-bound design condition and, consequently, an explicit admissible family of tuning functions that guarantees input compatibility. Based on this, we specialize the characterization to general constraint families, including norm-bounded, polyhedral, and box sets, to provide directly checkable tuning conditions for different input constraints. We then develop a tractable offline synthesis method for a parameterized tuning function, leveraging a covering-based sampling strategy to facilitate a linear program (LP)-based offline selection procedure. This approach upgrades pointwise compatibility to a formal compatibility guarantee over the entire safe set and ensures the recursive feasibility of the online QP-based safety filter throughout the safe set via the designed tuning family. Finally, we validate the framework through a practical connected cruise control (CCC) application, demonstrating that the proposed method preserves the robust safety benefits of TISSf while satisfying input constraints by design.

\vspace{-10pt}
\section{Preliminaries and Problem Formulation}\label{sec:prelim_problem}

Consider the control-affine system
\begin{equation}\label{Affine_Control_System}
    \dot{\mathbf{x}}=\mathbf{f}(\mathbf{x})+\mathbf{g}(\mathbf{x})\mathbf{u},
\end{equation}
where $\mathbf{x}\in\mathbb{R}^{n}$ is the state, $\mathbf{u}\in\mathbb{R}^{m}$ is the control input, and the vector fields
$\mathbf{f}:\mathbb{R}^{n}\rightarrow\mathbb{R}^{n}$ and $\mathbf{g}:\mathbb{R}^{n}\rightarrow\mathbb{R}^{n\times m}$ are locally Lipschitz.
Given a Lipschitz continuous state-feedback controller $\mathbf{k}:\mathbb{R}^{n}\rightarrow\mathbb{R}^{m}$, the closed-loop system is
\begin{equation}\label{state_feedback_dynamics}
\dot{\mathbf{x}} = \mathbf{f}_{\mathrm{cl}}(\mathbf{x})\triangleq\mathbf{f}(\mathbf{x})+\mathbf{g}(\mathbf{x}) \mathbf{k}(\mathbf{x}).
\end{equation}
Since $\mathbf{f}$, $\mathbf{g}$, and $\mathbf{k}$ are locally Lipschitz, so is $\mathbf{f}_{\mathrm{cl}}$.
Consequently, for any initial condition $\mathbf{x}_{0}\in\mathbb{R}^{n}$, there exists a maximal interval
$I(\mathbf{x}_{0})=[0,t_{\max})$ on which the solution to~\eqref{state_feedback_dynamics} exists and is unique.

\subsection{Control Barrier Functions}
Consider a closed set $\mathcal{C}\subset\mathbb{R}^{n}$ defined as the $0$-superlevel set of a continuously differentiable function
$h:\mathbb{R}^{n}\rightarrow\mathbb{R}$:
\begin{equation}\label{Invariant_Set}
\begin{aligned}
\mathcal{C} & \triangleq\{\mathbf{x}\in \mathbb{R}^{n}: h(\mathbf{x}) \geq 0\}, \\
\partial \mathcal{C} & \triangleq\{\mathbf{x}\in\mathbb{R}^{n}: h(\mathbf{x})=0\}, \\
\operatorname{Int}(\mathcal{C}) & \triangleq\{\mathbf{x}\in \mathbb{R}^{n}: h(\mathbf{x})>0\},
\end{aligned}
\end{equation}
where we assume that $\mathcal{C}$ is nonempty and has no isolated points.
\begin{Def}
(Forward Invariance $\&$ Safety) A set $\mathcal{C}\subset\mathbb{R}^{n}$ is \emph{forward invariant} if for every $\mathbf{x}_{0}\in\mathcal{C}$,
the solution to~\eqref{Affine_Control_System} satisfies $\mathbf{x}(t)\in\mathcal{C}$ for all $t\in I(\mathbf{x}_{0})$.
The system is \emph{safe} on $\mathcal{C}$ if $\mathcal{C}$ is forward invariant.
\end{Def}

\begin{Def}
(Class $\mathcal{K}_{\infty}$ and Extended Class $\mathcal{K}$ Functions) 
A continuous function $\alpha: [0, \infty) \rightarrow [0, \infty)$ is said to belong to class $\mathcal{K}_{\infty}$ if $\alpha(0) = 0$, $\alpha$ is strictly monotonically increasing, and $\lim_{r \to \infty} \alpha(r) = \infty$. Furthermore, a continuous function $\alpha:(-b,a)\rightarrow\mathbb{R}$, with $a,b>0$, is an extended class $\mathcal{K}$ function ($\alpha\in\mathcal{K}_{e}$) if $\alpha(0)=0$, and $\alpha$ is strictly monotonically increasing.
\end{Def}

\begin{Def}\label{CBF_Def}
(CBFs~\cite{Zeroing_CBF}) Let the set $\mathcal{C}$ be defined as the $0$-superlevel set of a continuously differentiable function $h: \mathbb{R}^n \rightarrow \mathbb{R}$ given in \eqref{Invariant_Set} with $\frac{\partial h}{\partial\mathbf{x}}(\mathbf{x})\neq\mathbf{0}$ when $h(\mathbf{x})=0$.
Then $h$ is a CBF for \eqref{Affine_Control_System} if there exists $\alpha\in\mathcal{K}_{e}$ such that, for all $\mathbf{x}\in\mathbb{R}^{n}$, there exists $\mathbf{u}\in\mathbb{R}^{m}$ satisfying
\begin{equation}\label{CBF_Condition}
c(\mathbf{x})+\mathbf{d}(\mathbf{x})\mathbf{u}\geq 0,
\end{equation}
where $c(\mathbf{x})=L_{\mathbf{f}} h(\mathbf{x})+\alpha(h(\mathbf{x}))$ and $\mathbf{d}(\mathbf{x})=L_{\mathbf{g}} h(\mathbf{x})$, $L_{\mathbf{f}} h(\mathbf{x})\triangleq\frac{\partial h(\mathbf{x})}{\partial\mathbf{x}}\mathbf{f}(\mathbf{x})$ and $L_{\mathbf{g}} h(\mathbf{x})\triangleq\frac{\partial h(\mathbf{x})}{\partial\mathbf{x}}\mathbf{g}(\mathbf{x})$ denote the Lie derivatives along $\mathbf{f}$ and $\mathbf{g}$, respectively. 
\end{Def}

Given a CBF $h$ and a corresponding $\alpha$, define the admissible set $\Phi(\mathbf{x})\triangleq\left\{\mathbf{u}\in\mathbb{R}^{m}\mid c(\mathbf{x})+\mathbf{d}(\mathbf{x})\mathbf{u}\geq 0\right\}$.

\begin{Thm}(\cite{Zeroing_CBF})
If $h$ is a CBF for~\eqref{Affine_Control_System} on $\mathcal{C}$, then any Lipschitz continuous controller satisfying
$\mathbf{k}(\mathbf{x})\in\Phi(\mathbf{x})$ for all $\mathbf{x}\in\mathcal{C}$ renders~\eqref{state_feedback_dynamics} safe with respect to $\mathcal{C}$.
\end{Thm}

\subsection{Tunable Input-to-State Safety}
Consider the dynamical system:
\begin{equation}\label{Affine_Control_Disturbed}
\dot{\mathbf{x}} = \mathbf{f}(\mathbf{x}) + \mathbf{g}(\mathbf{x})\big(\mathbf{u} + \bm{\omega}(t)\big),
\end{equation}
where $\bm{\omega}(t)$ denotes an external perturbation satisfying $\|\bm{\omega}\|_{\infty}\leq\delta$ for a known constant $\delta\geq 0$. As shown in~\cite{kolathaya2018input,alan2021safe}, a state-feedback law $\mathbf{k}(\mathbf{x}) \in \Phi(\mathbf{x})$ designed for the nominal system generally fails to guarantee safety in the presence of such perturbations.

Define the tightened set
\begin{equation}\label{Invariant_Set_Disturbed}
\begin{aligned}
\mathcal{C}_{\bm{\omega}} & \triangleq\left\{\mathbf{x} \in \mathbb{R}^{n}: h(\mathbf{x})+\zeta(h(\mathbf{x}),\delta) \geq 0\right\}, \\
\partial \mathcal{C}_{\bm{\omega}} & \triangleq\left\{\mathbf{x} \in \mathbb{R}^{n}: h(\mathbf{x})+\zeta(h(\mathbf{x}),\delta) = 0\right\}, \\
\operatorname{Int}(\mathcal{C}_{\bm{\omega}}) & \triangleq\left\{\mathbf{x} \in \mathbb{R}^{n}: h(\mathbf{x})+\zeta(h(\mathbf{x}),\delta)>0\right\},
\end{aligned}
\end{equation}
where $\zeta(\cdot,\delta)$ is continuously differentiable in its first argument and belongs to $\mathcal{K}_{\infty}$ in its second argument.

\begin{Def}
(TISSf-CBFs~\cite{alan2021safe}) Let the set $\mathcal{C}$ be defined as the $0$-superlevel set of a continuously differentiable function $h: \mathbb{R}^n \rightarrow \mathbb{R}$ given in \eqref{Invariant_Set} with $\frac{\partial h}{\partial \mathbf{x}}(\mathbf{x}) \neq \mathbf{0}$ whenever $h(\mathbf{x}) = 0$.
Then $h$ is a TISSf-CBF for \eqref{Affine_Control_System} on $\mathcal{C}$ with a continuously differentiable $\varepsilon:\mathbb{R}\rightarrow\mathbb{R}_{>0}$ if there exists $\alpha\in\mathcal{K}_{e}$ such that, for all $\mathbf{x}\in\mathbb{R}^{n}$, there exists $\mathbf{u}\in\mathbb{R}^{m}$ satisfying
\begin{equation}\label{HOCBF_ISSf}
c(\mathbf{x})+\mathbf{d}(\mathbf{x})\mathbf{u}\geq \frac{\left\|\mathbf{d}(\mathbf{x})\right\|_{2}^{2}}{\varepsilon(h(\mathbf{x}))}.
\end{equation}
\end{Def}

Given a tuning function $\varepsilon:\mathbb{R}\rightarrow\mathbb{R}_{>0}$, define the pointwise admissible set
\begin{equation}
\Pi_{\varepsilon}(\mathbf{x})\triangleq\left\{\mathbf{u}\in\mathbb{R}^{m}\mid c(\mathbf{x})+\mathbf{d}(\mathbf{x})\mathbf{u}\geq \frac{\left\|\mathbf{d}(\mathbf{x})\right\|_{2}^{2}}{\varepsilon(h(\mathbf{x}))}\right\}.
\end{equation}

As illustrated in Fig.~\ref{Geometric_inllustration} (a), for a fixed state $\mathbf{x}\in\mathbb{R}^{n}$, the TISSf-CBF condition \eqref{HOCBF_ISSf} defines a half-space $\Pi_{\varepsilon}(\mathbf{x})$ in the input space. In \eqref{HOCBF_ISSf}, the term $\|\mathbf{d}(\mathbf{x})\|_{2}^{2}/\varepsilon(h(\mathbf{x}))$ grows as the tuning function $\varepsilon(h(\mathbf{x}))$ decreases. This increases conservatism and shrinks the admissible control set $\Pi_{\varepsilon}(\mathbf{x})$ (orange dashed line) relative to the baseline (blue solid line). Conversely, increasing $\varepsilon(h(\mathbf{x}))$ makes \eqref{HOCBF_ISSf} less conservative and expands $\Pi_{\varepsilon}(\mathbf{x})$ (green dashed line). However, to ensure the TISSf property as established in \cite{alan2021safe}, the tuning function $\varepsilon(h(\mathbf{x}))$ must also satisfy the monotonicity condition \eqref{TISSf_Function_Limit}, as formalized in the following theorem.
\begin{figure}[tp]
 \centering
 \includegraphics[width=0.85\columnwidth]{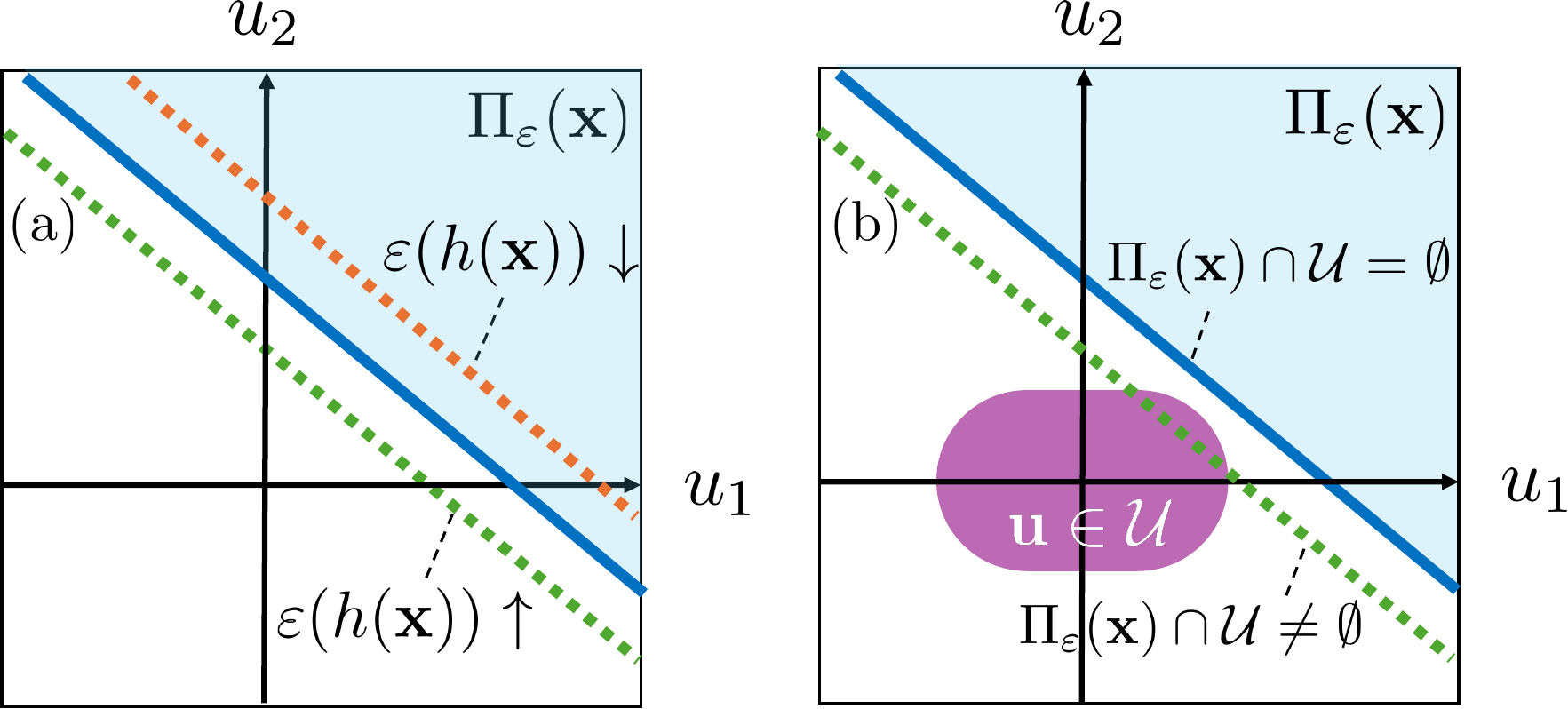}
 \caption{Geometric interpretation of the TISSf-CBF condition and its compatibility with input constraints. 
 (a) Geometric effect of $\varepsilon(h(\mathbf{x}))$: varying its value shifts the half-space boundary relative to a baseline (blue line), where increasing or decreasing $\varepsilon(h(\mathbf{x}))$ yields a displaced hyperplane (green and orange dashed lines, respectively). 
 (b) Design impact on compatibility: an improperly small $\varepsilon(h(\mathbf{x}))$ results in $\Pi_{\varepsilon}(\mathbf{x}) \cap \mathcal{U} = \emptyset$ (blue line), whereas a proper choice of $\varepsilon(h(\mathbf{x}))$ ensures $\Pi_{\varepsilon}(\mathbf{x}) \cap \mathcal{U} \neq \emptyset$ (green dashed line).}
 \label{Geometric_inllustration}
\end{figure}

\begin{Thm}\label{ISSf_Theory}
(\cite{alan2021safe}) Let the set $\mathcal{C}$ be defined as the $0$-superlevel set of a continuously differentiable function $h: \mathbb{R}^n \rightarrow \mathbb{R}$ given in \eqref{Invariant_Set} with $\frac{\partial h}{\partial \mathbf{x}}(\mathbf{x}) \neq \mathbf{0}$ whenever $h(\mathbf{x}) = 0$ and $\delta\geq 0$.
Assume that $\varepsilon:\mathbb{R}\rightarrow\mathbb{R}_{> 0}$ is continuously differentiable and satisfies
\begin{equation}\label{TISSf_Function_Limit}
    \frac{\mathrm{d}\varepsilon(h(\mathbf{x}))}{\mathrm{d}h(\mathbf{x})}\geq 0.
\end{equation}
If $h$ is a TISSf-CBF for~\eqref{Affine_Control_System} on $\mathcal{C}$, then any Lipschitz continuous controller with $\mathbf{k}(\mathbf{x})\in\Pi_{\varepsilon}(\mathbf{x})$ for all $\mathbf{x}\in\mathcal{C}$
renders~\eqref{Affine_Control_Disturbed} safe with respect to $\mathcal{C}_{\bm{\omega}}$, where in~\eqref{Invariant_Set_Disturbed} we have
$\zeta(h(\mathbf{x}),\delta)=-\alpha^{-1}(-\frac{\varepsilon (h(\mathbf{x}))\delta^{2}}{4})$.
\end{Thm}

\subsection{Problem Formulation}
The TISSf framework in Theorem~\ref{ISSf_Theory} provides a tuning function $\varepsilon(h(\mathbf{x}))$ as a design degree of freedom, but it does not account for input constraints. To address this limitation, we consider a general input constraint
\begin{equation}\label{general_constraint}
\mathbf{u}\in\mathcal{U},
\end{equation}
where $\mathcal{U}\subset\mathbb{R}^{m}$ is nonempty, compact, and convex. As shown in Fig.~\ref{Geometric_inllustration} (b), the TISSf-CBF condition \eqref{HOCBF_ISSf} defines a half-space $\Pi_{\varepsilon}(\mathbf{x})$ in the input space whose boundary shifts based on the value of $\varepsilon(h(\mathbf{x}))$. For a fixed state $\mathbf{x}$, an improperly small $\varepsilon(h(\mathbf{x}))$ results in an empty intersection $\Pi_{\varepsilon}(\mathbf{x})\cap\mathcal{U} = \emptyset$ as represented by the blue line, leading to a loss of compatibility. Increasing the value of $\varepsilon(h(\mathbf{x}))$ translates this boundary toward the green dashed line, expanding the half-space to ensure $\Pi_{\varepsilon}(\mathbf{x})\cap\mathcal{U} \neq \emptyset$. This reveals that incorporating input constraints into the control design does not merely add a posterior compatibility check but imposes a fundamental restriction on the admissible family of tuning functions. This insight motivates characterizing the set of tuning functions that ensure $\Pi_{\varepsilon}(\mathbf{x})\cap\mathcal{U}\neq\emptyset$ for all $\mathbf{x}\in\mathcal{C}$ and deriving corresponding conditions for common constraint sets, which is further illustrated by the following example adapted from~\cite{alan2021safe}.
\begin{Exam}\label{Typical_Example}
Consider the system $\dot{x}_{1}=-x_{2}, \dot{x}_{2}=u+\omega(t)$, where $\mathbf{x}=[x_{1},x_{2}]^{\top}$, $u\in\mathbb{R}$, and $\omega(t)=\mu\sin t$ with $\mu=3$.
Following~\cite{alan2021safe}, choose the TISSf-CBF candidate $h(\mathbf{x})=x_{1}-x_{2}$ and the nominal feedback $k(\mathbf{x})=x_{1}-2x_{2}-1$.
We consider the TISSf-CBF-based control law used in~\cite{alan2021safe}, $ u = k(\mathbf{x})+d(\mathbf{x})/\varepsilon(h(\mathbf{x}))$,
where $d(\mathbf{x})=L_{\mathbf{g}}h(\mathbf{x})=-1$ and $\varepsilon(h(\mathbf{x}))=\epsilon_{0}e^{\lambda h(\mathbf{x})}$.
In the simulations, we fix $\lambda=0.2$ and vary $\ln\epsilon_{0}\in\{-4, -3.5, -2, 0, 1\}$.
The input bound is $|u| \le 15$, which can be characterized as either a box constraint or a norm-bounded constraint.
\end{Exam}
Fig.~\ref{Input_Behavior} illustrates that certain tuning choices may satisfy the TISSf condition for the perturbed system~\eqref{Affine_Control_Disturbed} while violating the prescribed input bounds. Specifically, different parameter choices in $\varepsilon(h(\mathbf{x}))=\epsilon_0 e^{\lambda h(\mathbf{x})}$ lead to different closed-loop trajectories (Fig.~\ref{Input_Behavior} (a)).
Although the nominal constraint $h(\mathbf{x})\ge 0$ may be violated (Fig.~\ref{Input_Behavior} (b)), the tightened inequality $h(\mathbf{x})+\zeta(h(\mathbf{x}),\delta)\ge 0$ can remain satisfied (Fig.~\ref{Input_Behavior} (c)), which corresponds to safety with respect to the tightened set $\mathcal{C}_{\bm{\omega}}$ in Theorem~\ref{ISSf_Theory}.
However, the resulting control input may exceed the input constraint with certain choices of parameter pairs (Fig.~\ref{Input_Behavior}(d)).
This reveals a fundamental design question: \textit{In the presence of input constraints, how can one construct a tuning function $\varepsilon(\cdot)$ that preserves TISSf while ensuring $\Pi_{\varepsilon}(\mathbf{x}) \cap \mathcal{U} \neq \emptyset$ for all $\mathbf{x} \in \mathcal{C}$?} Equivalently, we seek parameterizations $(\ln \epsilon_0, \lambda)$ that maintain both TISSf and input compatibility over the entire safe set $\mathcal{C}$. These considerations lead to the following problem formulation.
\begin{figure}[tp]
 \centering
 \includegraphics[width=1\columnwidth]{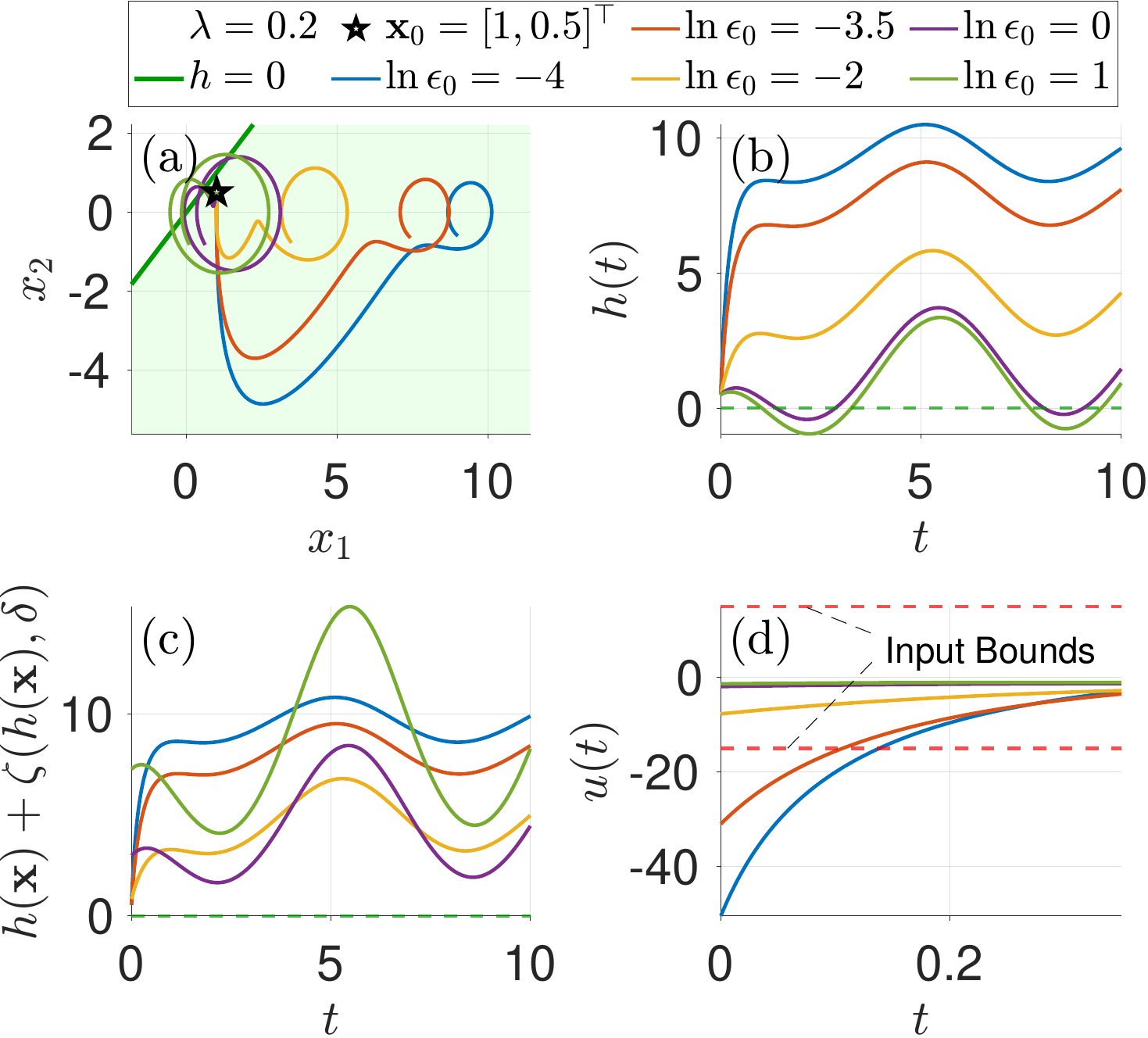}
 \caption{Motivating example showing that directly applying a TISSf-CBF-based controller with different parameters of the exponential tuning function $\varepsilon(h)=\epsilon_0 e^{\lambda h}$ may violate input constraints:
(a) state trajectories in the $(x_1,x_2)$ plane for different parameter choices;
(b) evolution of the barrier function $h(\mathbf x)$;
(c) evolution of the tightened quantity $h(\mathbf x)+\zeta(h(\mathbf x),\delta)$;
(d) control input $u(t)$ over a short time window, highlighting possible violations of the input bound $|u|\le 15$.}
 \label{Input_Behavior}
\end{figure}
\begin{Prob}\label{Problem_Formulation}
Consider the perturbed affine system \eqref{Affine_Control_Disturbed} with a nonempty compact admissible input set $\mathcal U\subset\mathbb{R}^{m}$.
Let $\mathcal C=\{\mathbf x\in\mathbb{R}^n:\,h(\mathbf x)\ge 0\}$ and assume that $h$ is a TISSf-CBF for \eqref{Affine_Control_System} on $\mathcal C$.
The objective is to: 
(i) characterize the set of admissible tuning functions
\begin{equation}
\mathcal{E}:=\Big\{\varepsilon:\mathbb{R}\rightarrow\mathbb{R}_{>0}\ \Big|\ \Pi_{\varepsilon}(\mathbf x)\cap \mathcal U \neq \emptyset,\ \forall \mathbf x\in\mathcal C \Big\};
\end{equation}
(ii) derive explicit tuning conditions for commonly used input sets $\mathcal U$; and 
(iii) establish parameter design conditions for structured choices of $\varepsilon(\cdot)$ so that $\varepsilon(\cdot)$ remains a design degree of freedom under the constraint $\mathbf u\in\mathcal U$.
\end{Prob}

\section{TISSf with Input Constraints}\label{TISSf_with_Input_Limits}
In this section, we address Problem~\ref{Problem_Formulation} by deriving an exact support function characterization of tuning compatibility under a general compact input constraint. We then discuss several common input constraints, including norm-bounded, polyhedral, and box input constraints, and develop a tractable parametric tuning method based on exponential parameterization together with a covering-based sampling strategy that leads to an LP formulation for offline parameter design.

\subsection{Characterize the Set of Tuning Functions }\label{Main_A}
In this subsection, we consider the general compact input constraint $u \in \mathcal{U}$ defined in \eqref{general_constraint}. By transforming the compatibility requirement $\Pi_{\varepsilon}(\mathbf{x}) \cap \mathcal{U} \neq \emptyset$ into an explicit lower bound on $\varepsilon(\cdot)$, we treat input compatibility as a primary design objective. More specifically, rather than checking compatibility a posteriori, we proactively select $\varepsilon(\cdot)$ to shift the TISSf-CBF half-space and guarantee the intersection condition by design. This approach relies on an exact scalar characterization of the compatibility condition via a support function, yielding a verifiable, closed-form bound on $\varepsilon(h(\mathbf{x}))$.
\begin{Lem}\label{lem:general_compatibility}
For a given $\mathbf{x}\in\mathbb{R}^n$, assume $c(\mathbf{x})+\sigma_{\mathcal{U}}(\mathbf{d}(\mathbf{x})) > 0$, where $\sigma_{\mathcal{U}}(\mathbf{d}(\mathbf{x})) \triangleq \max_{\mathbf{u}\in\mathcal{U}} \mathbf{d}(\mathbf{x})\mathbf{u}$ denotes the support function of $\mathcal{U}$. Then, the TISSf-CBF constraint \eqref{HOCBF_ISSf} is compatible with the input constraint $\mathcal{U}$, i.e., $\Pi_{\varepsilon}(\mathbf{x})\cap\mathcal{U}\neq\emptyset$, if and only if 
\begin{equation}\label{eq:general_eps_bound}
\varepsilon(h(\mathbf{x})) \geq \frac{\|\mathbf{d}(\mathbf{x})\|_{2}^2}{c(\mathbf{x})+\sigma_{\mathcal{U}}(\mathbf{d}(\mathbf{x}))}.
\end{equation}
\end{Lem}
\begin{proof}
As shown in Fig.~\ref{Geometric_inllustration} (b), the TISSf condition \eqref{HOCBF_ISSf} defines a state-dependent half-space $\Pi_{\varepsilon}(\mathbf{x})$ in the input space, which can be rearranged as:
\begin{equation}\label{eq:halfspace_form}
    \frac{\|\mathbf{d}(\mathbf{x})\|_{2}^{2}}{\varepsilon(h(\mathbf{x}))} - c(\mathbf{x}) \leq \mathbf{d}(\mathbf{x})\mathbf{u}.
\end{equation}
Necessity follows from the fact that compatibility requires the existence of at least one $\mathbf{u} \in \mathcal{U}$ satisfying \eqref{eq:halfspace_form}, which implies:
\begin{equation}\label{itermidiarte}
    \frac{\|\mathbf{d}(\mathbf{x})\|_{2}^{2}}{\varepsilon(h(\mathbf{x}))} - c(\mathbf{x}) \leq \max_{\mathbf{u} \in \mathcal{U}} \mathbf{d}(\mathbf{x}) \mathbf{u} = \sigma_{\mathcal{U}}(\mathbf{d}(\mathbf{x})).
\end{equation}
Under the assumption $c(\mathbf{x}) + \sigma_{\mathcal{U}}(\mathbf{d}(\mathbf{x})) > 0$, Eq.~\eqref{itermidiarte} is equivalent to \eqref{eq:general_eps_bound}. For sufficiency, assume \eqref{eq:general_eps_bound} holds. Since $\mathcal{U}$ is non-empty and compact, there exists $\mathbf{u}^\star \in \mathcal{U}$ such that $\mathbf{d}(\mathbf{x}) \mathbf{u}^\star = \sigma_{\mathcal{U}}(\mathbf{d}(\mathbf{x}))$. Substituting this into \eqref{eq:general_eps_bound} produces \eqref{eq:halfspace_form}, which further shows $\mathbf{u}^\star \in \mathcal{U} \cap \Pi_{\varepsilon}(\mathbf{x})$ and completes the proof.
\end{proof}

\begin{Rmk}\label{rem:well_defined_bound}
The condition $c(\mathbf{x})+\sigma_{\mathcal{U}}(\mathbf{d}(\mathbf{x}))>0$ in Lemma~\ref{lem:general_compatibility} is a necessary and unavoidable requirement for enforcing TISSf under input constraints. Recall $\sigma_{\mathcal{U}}(\mathbf{d}(\mathbf{x})) \triangleq \max_{\mathbf{u}\in\mathcal{U}} \mathbf{d}(\mathbf{x})\mathbf{u}$. When $\varepsilon(h(\mathbf{x}))\to\infty$, the TISSf condition~\eqref{HOCBF_ISSf} reduces to the nominal CBF half-space $\Pi_{\infty}(\mathbf{x})\!:=\!\{\mathbf{u}:\mathbf{d}(\mathbf{x})\mathbf{u}\ge -c(\mathbf{x})\}$ (cf.~\eqref{CBF_Condition}), which is compatible with $\mathcal{U}$ if and only if $c(\mathbf{x})+\sigma_{\mathcal{U}}(\mathbf{d}(\mathbf{x}))\ge 0$. If $c(\mathbf{x})+\sigma_{\mathcal{U}}(\mathbf{d}(\mathbf{x}))< 0$, then $\Pi_{\infty}(\mathbf{x})\cap\mathcal{U}=\emptyset$ and hence no tuning can recover compatibility: for any finite $\varepsilon(h(\mathbf{x}))$, the TISSf admissible set $\Pi_{\varepsilon}(\mathbf{x})=\{\mathbf{u}:\mathbf{d}(\mathbf{x})\mathbf{u}\ge \|\mathbf{d}(\mathbf{x})\|^2/\varepsilon(h(\mathbf{x}))-c(\mathbf{x})\}$ satisfies $\Pi_{\varepsilon}(\mathbf{x})\subseteq \Pi_{\infty}(\mathbf{x})$. Furthermore, Lemma~\ref{lem:general_compatibility} is stated with $c(\mathbf{x})+\sigma_{\mathcal{U}}(\mathbf{d}(\mathbf{x}))>0$ to exclude the boundary case $c(\mathbf{x})+\sigma_{\mathcal{U}}(\mathbf{d}(\mathbf{x}))=0$, where the lower bound in~\eqref{eq:general_eps_bound} becomes singular (i.e., the right-hand side becomes unbounded). Therefore, throughout the paper we restrict attention to the (nominally) compatible regime $c(\mathbf{x})+\sigma_{\mathcal{U}}(\mathbf{d}(\mathbf{x}))>0$.
\end{Rmk}

Lemma~\ref{lem:general_compatibility} establishes the pointwise compatibility of the TISSf-CBF constraint at a fixed state $\mathbf{x}$. 
However, in practice, a tuning function $\varepsilon(\cdot)$ is typically selected offline through fixed parameterizations and must remain valid across the entire safe set $\mathcal{C}$. As shown in Example~\ref{Typical_Example}, a tuning function may satisfy the input constraints at certain states but fail to guarantee such satisfaction throughout the trajectory (e.g., the case $\ln\epsilon_{0}=-4$ in Fig. 2 (d)).  Therefore, lifting this condition from a fixed state $\mathbf{x}$ to the entire safe set $\mathcal{C}$ is a prerequisite to characterizing the functional space $\mathcal{E}$ required by Problem~\ref{Problem_Formulation} and to ensuring that the selected $\varepsilon(\cdot)$ remains implementable (i.e., $\Pi_{\varepsilon}(\mathbf x)\cap \mathcal U \neq \emptyset$) along any feasible closed-loop trajectory.

\begin{Cor}\label{global_compatibility}
By Lemma~\ref{lem:general_compatibility}, the set $\mathcal{E}$ defined in Problem~\ref{Problem_Formulation} admits the characterization
\begin{equation}\label{compatible_chara}
\mathcal{E} = \left\{ \varepsilon:\mathbb{R}\to\mathbb{R}_{>0} \ \middle|\ \eqref{eq:general_eps_bound} \text{ holds } \forall \mathbf{x}\in\mathcal{C} \right\}.
\end{equation}
\end{Cor}


In the following, we combine Theorem~\ref{ISSf_Theory} with Corollary~\ref{global_compatibility} to establish the existence of a Lipschitz controller $\mathbf{k}$ satisfying $\mathbf{k}(\mathbf{x})\in\Pi_{\varepsilon}(\mathbf{x})\cap\mathcal{U}$ for all $\mathbf{x}\in\mathcal{C}$, and consequently the safety of the perturbed system with respect to $\mathcal{C}_{\bm{\omega}}$, as stated next.

\begin{Thm}\label{thm:input_constrained_tissf}
Let $\mathcal{C}\subset\mathbb{R}^{n}$ be the $0$-superlevel set of a continuously differentiable function $h:\mathbb{R}^{n}\rightarrow\mathbb{R}$ and $\delta\geq 0$.
Let the assumptions of Theorem~\ref{ISSf_Theory} hold and let $\mathcal U$ be a nonempty compact input set. 
Suppose the tuning function $\varepsilon: \mathbb{R} \to \mathbb{R}_{>0}$ is continuously differentiable and satisfies the monotonicity condition~\eqref{TISSf_Function_Limit} and $\varepsilon(\cdot)\in\mathcal{E}$.
If there exists a Lipschitz continuous controller $\mathbf{k}(\mathbf{x})$ satisfying
$\mathbf{k}(\mathbf{x})\in \Pi_{\varepsilon}(\mathbf{x})\cap \mathcal U,  \forall\,\mathbf{x}\in\mathcal C$,
then the perturbed system~\eqref{Affine_Control_Disturbed} is safe with respect to $\mathcal C_{\bm{\omega}}$ with the same robustness margin $\zeta(h(\mathbf{x}),\delta)$ as in Theorem~\ref{ISSf_Theory}.
\end{Thm}

\begin{proof}
Since $\varepsilon(\cdot)\in\mathcal{E}$, Corollary~\ref{global_compatibility} guarantees that $\Pi_{\varepsilon}(\mathbf{x})\cap\mathcal U\neq\emptyset$ for all $\mathbf{x}\in\mathcal C$.
By assumption, there exists a Lipschitz continuous controller $\mathbf{k}(\mathbf{x})$ such that
$\mathbf{k}(\mathbf{x})\in \Pi_{\varepsilon}(\mathbf{x})\cap\mathcal U$ for all $\mathbf{x}\in\mathcal C$. In particular, $\mathbf{k}(\mathbf{x})\in \Pi_{\varepsilon}(\mathbf{x})$ for all $\mathbf{x}\in\mathcal C$.
Therefore, by utilizing Theorem~\ref{ISSf_Theory}, safety with respect to $\mathcal C_{\bm{\omega}}$ and the robustness margin $\zeta(h(\mathbf{x}),\delta)$ follows directly from Theorem~\ref{ISSf_Theory}. 
\end{proof}

\subsection{Tuning Conditions for Common Input Sets}\label{Main_B}
Next, we make the compatibility condition \eqref{eq:general_eps_bound} (and correspondingly the set characterization in \eqref{compatible_chara}) explicit for several common input constraint sets. In the following, we discuss norm-bounded, polyhedral, and box input constraints.
\subsubsection{Norm-Bounded Input Constraints}
We first consider the Euclidean norm-bounded admissible control set
\begin{equation}\label{eq:control_input_limit}
    \mathcal{U}=\{\mathbf{u}\in\mathbb{R}^{m}\mid\|\mathbf{u}\|_{2}\leq \gamma\},
\end{equation}
where $\gamma>0$ is constant. For \eqref{eq:control_input_limit}, the support function is
\begin{equation}\label{eq:support_norm}
    \sigma_{\mathcal{U}}(\mathbf{d}(\mathbf{x}))
    =
    \max_{\|\mathbf{u}\|_{2}\le \gamma}\mathbf{d}(\mathbf{x})^\top\mathbf{u}
    =
    \gamma\|\mathbf{d}(\mathbf{x})\|_{2}.
\end{equation}
Substituting \eqref{eq:support_norm} into \eqref{eq:general_eps_bound} yields the explicit lower bound. Therefore, Theorem~\ref{thm:input_constrained_tissf} guarantees the existence of a
controller satisfying both the TISSf-CBF condition and the norm-bounded input constraint.

\subsubsection{Polyhedral Input Constraints}
We next consider polyhedral admissible control sets
\begin{equation}\label{eq:polyhedral_constraint}
    \mathcal{U} = \{\mathbf{u}\in\mathbb{R}^{m} \mid \mathbf{A}\mathbf{u} \le \mathbf{b}\},
\end{equation}
with $\mathbf{A}\in\mathbb{R}^{p\times m}$, $\mathbf{b}\in\mathbb{R}^{p}$, and $\mathcal{U}$
nonempty and compact. The support function is the optimal value of a LP and admits the dual form
\begin{equation}\label{eq:support_polyhedral}
    \sigma_{\mathcal{U}}(\mathbf{d}(\mathbf{x}))
    =
    \max_{\mathbf{A}\mathbf u\le \mathbf b}\mathbf{d}(\mathbf{x})^\top \mathbf{u}
    =
    \min_{\substack{\boldsymbol{\lambda}\ge \mathbf{0}\\ \mathbf{A}^\top\boldsymbol{\lambda}= \mathbf{d}(\mathbf{x})}}
    \boldsymbol{\lambda}^\top\mathbf{b},
\end{equation}
where strong duality holds since $\mathcal{U}$ is assumed to be nonempty and compact, so the primal LP is feasible and attains a finite optimum. In general, evaluating \eqref{eq:support_polyhedral} may require an LP solver (or an
explicit multiparametric solution), but Theorem~\ref{thm:input_constrained_tissf}
remains valid  for any tuning function $\varepsilon(\cdot)$ such that it belongs to $\mathcal{E}$ with \(\sigma_{\mathcal U}\) given by \eqref{eq:support_polyhedral}.

\subsubsection{Box Input Constraints}
As a special case of \eqref{eq:polyhedral_constraint} with $\mathbf{A}$ and $\mathbf{b}$ chosen to encode componentwise bounds, we obtain box input constraints as follows.
\begin{equation}\label{eq:box_constraint}
    \mathcal{U}=\{\mathbf{u}\in\mathbb{R}^{m}\mid\mathbf{u}_{\min}\leq\mathbf{u}\leq\mathbf{u}_{\max}\},
\end{equation}
with constant componentwise bounds \(\mathbf u_{\min},\mathbf u_{\max}\).
Then the support function decouples componentwise:
\begin{equation}\label{eq:support_box}
\begin{split}
    \sigma_{\mathcal{U}}(\mathbf{d}(\mathbf{x}))
    & =
    \sum_{i=1}^{m}
    \max\{d_i(\mathbf{x})\,u_{\max,i},\ d_i(\mathbf{x})\,u_{\min,i}\}
\end{split}
\end{equation}
Substituting~\eqref{eq:support_box} into~\eqref{eq:general_eps_bound} yields an explicit lower bound on the tuning function $\varepsilon(\cdot)$. Note that, with the tuning function $\varepsilon$ designed to satisfy \eqref{eq:general_eps_bound} $\forall \mathbf{x}\in\mathcal{C}$, the specific control input $\mathbf{u}$ can then be computed by solving the QP~\eqref{eq:TISSf_QP} in Section~\ref{subsec:controller_synthesis}.

\subsection{Parametric Tuning Function}\label{subsec:param_tuning}
Section~\ref{Main_A} demonstrates that input constraints do not merely serve as a posteriori compatibility checks. Instead, they restrict the admissible tuning function $\varepsilon(\cdot)$ by requiring $\Pi_{\varepsilon}(\mathbf{x})\cap\mathcal{U}\neq\emptyset$ for all $\mathbf{x}\in\mathcal{C}$. 
Consequently, by designing $\varepsilon(\cdot)$ offline such that it belongs to the function space $\mathcal{E}$ defined in \eqref{compatible_chara}, we guarantee the TISSf of system~\eqref{Affine_Control_Disturbed} while strictly respecting the prescribed input constraints. To transform this infinite-dimensional functional design problem into a tractable finite-dimensional parameter design problem, we proceed by parameterizing $\varepsilon(\cdot)$ and recasting the compatibility requirement as explicit constraints on the design parameters.
\subsubsection{Parametric tuning function}
Following~\cite{alan2021safe}, we parameterize the tuning function $\varepsilon(\cdot)$ as
\begin{equation}\label{eq:exp_eps_def}
    \varepsilon(h(\mathbf{x})) \triangleq \epsilon_{0} e^{\lambda h(\mathbf{x})},
    \quad \epsilon_{0}>0,\ \lambda>0.
\end{equation}
This choice guarantees $\varepsilon(h(\mathbf{x}))>0$ and yields a smooth, strictly increasing dependence on the safety measure $h(\mathbf{x})$, which induces a state-dependent reduction of safety conservatism. 
Specifically, as the state moves away from the unsafe set, the increase in $h(\mathbf{x})$ raises the tuning value $\varepsilon(h(\mathbf{x}))$, thereby shrinking the robustness margin $\|\mathbf{d}(\mathbf{x})\|_{2}^2/\varepsilon(h(\mathbf{x}))$ and expanding the admissible input half-space $\Pi_{\varepsilon}(\mathbf{x})$. 
Furthermore, $\lambda$ characterizes the constant relaxation rate of this constraint, as $\frac{d}{dh}\log\varepsilon(h)=\lambda$ implies that each unit increase in the safety margin $h$ scales the admissible input flexibility by a factor of $e^\lambda$. In the following,  we will show that this two-parameter form allows the compatibility requirement to be expressed as affine inequalities in $(\ln\epsilon_0,\lambda)$, enabling systematic offline tuning and direct comparison with~\cite{alan2021safe}.
\begin{Cor}\label{lem:exp_param_feas_pointwise}
Fix any state $\mathbf{x}\in\mathbb{R}^n$ and suppose that
$c(\mathbf{x})+\sigma_{\mathcal{U}}(\mathbf{d}(\mathbf{x}))>0$ whenever $\|\mathbf{d}(\mathbf{x})\|_{2}>0$.
Let $\varepsilon(\cdot)$ be given by the exponential tuning function~\eqref{eq:exp_eps_def} and define
\begin{equation}\label{eq:eta_def_pw}
    \eta(\mathbf{x}) \triangleq \ln\|\mathbf{d}(\mathbf{x})\|_{2}^{2}
    - \ln\!\big(c(\mathbf{x})+\sigma_{\mathcal{U}}(\mathbf{d}(\mathbf{x}))\big).
\end{equation}
Then the TISSf-CBF constraint~\eqref{HOCBF_ISSf} is compatible with the input constraint~\eqref{general_constraint}
at $\mathbf{x}$, i.e., $\Pi_{\varepsilon}(\mathbf{x})\cap\mathcal{U}\neq\emptyset$, if and only if
\begin{equation}\label{eq:param_halfspace_pw}
    \ln\epsilon_{0} + \lambda h(\mathbf{x}) \;\ge\; \eta(\mathbf{x}).
\end{equation}
\end{Cor}
\begin{proof}
By Lemma~\ref{lem:general_compatibility}, the compatibility condition $\Pi_{\varepsilon}(\mathbf{x})\cap\mathcal{U}\neq\emptyset$ holds if and only if
\begin{equation}\label{eq:equiv_step}
    \epsilon_{0}e^{\lambda h(\mathbf{x})} \ge \frac{\|\mathbf{d}(\mathbf{x})\|_{2}^{2}}{c(\mathbf{x})+\sigma_{\mathcal{U}}(\mathbf{d}(\mathbf{x}))}.
\end{equation}
Under the assumption $c(\mathbf{x})+\sigma_{\mathcal{U}}(\mathbf{d}(\mathbf{x}))>0$ whenever $\|\mathbf{d}(\mathbf{x})\|_{2}>0$, both sides of \eqref{eq:equiv_step} are strictly positive. Taking the logarithm of both sides preserves the logical equivalence:
\begin{equation}
    \ln\epsilon_{0} + \lambda h(\mathbf{x}) \ge \ln\|\mathbf{d}(\mathbf{x})\|_{2}^2 - \ln\big(c(\mathbf{x})+\sigma_{\mathcal{U}}(\mathbf{d}(\mathbf{x}))\big).
\end{equation}
Substituting the definition of $\eta(\mathbf{x})$ from \eqref{eq:eta_def_pw} directly yields \eqref{eq:param_halfspace_pw}, which completes the proof.
\end{proof}
\begin{Rmk}\label{rem:eta_well_defined}
In Lemma~\ref{lem:exp_param_feas_pointwise}, to avoid numerical degeneracy when evaluating $\eta(\mathbf{x})$, we further assume that on the compact set $\mathcal{C}$ there exist constants $\underline d,\underline s>0$ such that $\|\mathbf{d}(\mathbf{x})\|\ge \underline d$ and $c(\mathbf{x})+\sigma_{\mathcal{U}}(\mathbf{d}(\mathbf{x}))\ge \underline s$ for all $\mathbf{x}\in\mathcal{C}$. Under these lower bounds, $\eta(\mathbf{x})$ is well-defined.
\end{Rmk}
Lemma~\ref{lem:exp_param_feas_pointwise} characterizes how input constraints restrict the choice of tuning parameters $(\ln\epsilon_0, \lambda)$. In the absence of input constraints, the exponential tuning function remains valid for any $\epsilon_0, \lambda > 0$, since Theorem~\ref{ISSf_Theory} only requires the satisfaction of the condition~\eqref{TISSf_Function_Limit}. However, once input constraints are introduced, Lemma~\ref{lem:general_compatibility} implies that the tuning function must satisfy the lower bound in \eqref{eq:general_eps_bound} for compatibility. At each state $\mathbf{x}$, this condition is equivalent to requiring the parameters to lie within a half-space defined by the affine inequality~\eqref{eq:param_halfspace_pw}. Similar to the transition to Corollary~\ref{global_compatibility}, ensuring admissibility for all system trajectories requires these parameters to satisfy the condition over the entire safe set $\mathcal{C}$, leading to the following corollary for the parameterized function space.
\begin{Cor}\label{cor:parameter_set}
For the exponential tuning function defined in \eqref{eq:exp_eps_def}, $\varepsilon(\cdot) \in \mathcal{E}$ if and only if the parameter pair $(\ln\epsilon_0, \lambda)$ satisfies
\begin{equation}\label{eq:global_param_constraint}
    \ln\epsilon_{0} + \lambda h(\mathbf{x}) \ge \eta(\mathbf{x}), \quad \forall \mathbf{x} \in \mathcal{C}.
\end{equation}
\end{Cor}
By satisfying the parametric constraint \eqref{eq:global_param_constraint}, the tuning function $\varepsilon(\cdot)$ belongs to the admissible space $\mathcal{E}$. Moreover, since $\varepsilon_{0}>0, \lambda > 0$ ensures the monotonicity condition \eqref{TISSf_Function_Limit}, this design constitutes a parametric realization of Theorem~\ref{thm:input_constrained_tissf}. 
Consequently, all conclusions of Theorem~\ref{thm:input_constrained_tissf} hold, which guarantees the existence of a Lipschitz controller $\mathbf{k}: \mathcal{C} \to \mathcal{U}$ that renders the perturbed system~\eqref{Affine_Control_Disturbed} TISSf.
\subsubsection{Parametric Tuning}\label{Pramter_tuning_method}
Corollary \ref{cor:parameter_set} characterizes the admissible tuning functions through an infinite family of inequalities indexed by $\mathbf{x} \in \mathcal{C}$. 
To obtain a tractable synthesis procedure, we replace this infinite-dimensional constraint with a finite set of constraints evaluated over a representative set of samples $\{\mathbf{x}_i\}_{i=1}^{N} \subset \mathcal{C}$. 
Specifically, we employ the concept of a $\kappa$-covering \cite{vershynin2018high} to discretize the safe set $\mathcal{C}$, ensuring that any point in the safe set $\mathcal{C}$ is within a distance $\kappa$ of at least one sample point. 
The following proposition establishes the theoretical guarantee for this sampling-based approach using Lipschitz continuity analysis.
\begin{Prop}\label{paramter_sample_condition}
Let $\mathcal{C}\subset\mathbb{R}^n$ be a compact set and let $\mathcal{X}_N$ be a $\kappa$-covering of $\mathcal{C}$, i.e., for any
$\mathbf{x}\in\mathcal{C}$ there exists $\mathbf{x}_i\in\mathcal{X}_N$
such that $\|\mathbf{x}-\mathbf{x}_i\|\le \kappa$, $\kappa\in\mathbb{R}_{>0}$.
Assume that $\eta(\cdot)$ and $h(\cdot)$ are Lipschitz continuous on $\mathcal{C}$
with constants $L_\eta>0$ and $L_h>0$, respectively.
If the parameter pair $(\ln\epsilon_{0}, \lambda)$ satisfies the robust sampled constraints
\begin{equation}\label{sampling_Lip}
\ln\epsilon_{0} + \lambda\big(h(\mathbf{x}_i)-L_h\kappa\big)
\ge
\eta(\mathbf{x}_i) + L_\eta\kappa,
\quad
\forall \mathbf{x}_i \in \mathcal{X}_N,
\end{equation}
then the tuning function $\varepsilon(h(\mathbf{x}))=\epsilon_{0} e^{\lambda h(\mathbf{x})}$ belongs to the admissible space $\mathcal{E}$,
and consequently $\Pi_{\varepsilon}(\mathbf{x}) \cap \mathcal{U} \neq \emptyset$ for all $\mathbf{x} \in \mathcal{C}$.
\end{Prop}

\begin{proof}
Fix any $\mathbf{x}\in\mathcal{C}$. Since $\mathcal{X}_N$ is a $\kappa$-covering of $\mathcal{C}$,
there exists $\mathbf{x}_i\in\mathcal{X}_N$ such that $\|\mathbf{x}-\mathbf{x}_i\|\le\kappa$.
By the Lipschitz continuity of $\eta(\cdot)$ and $h(\cdot)$ on $\mathcal{C}$, we have $\eta(\mathbf{x}) \le \eta(\mathbf{x}_i) + L_\eta\|\mathbf{x}-\mathbf{x}_i\|
\le \eta(\mathbf{x}_i) + L_\eta\kappa$ and $h(\mathbf{x}) \ge h(\mathbf{x}_i) - L_h\|\mathbf{x}-\mathbf{x}_i\|
\ge h(\mathbf{x}_i) - L_h\kappa$.
Therefore,
\begin{equation}
\begin{split}
    \ln\epsilon_0 + \lambda h(\mathbf{x})
&\ge
\ln\epsilon_0 + \lambda\big(h(\mathbf{x}_i)-L_h\kappa\big)\\
&\overset{\eqref{sampling_Lip}}{\ge}
\eta(\mathbf{x}_i)+L_\eta\kappa
\ge
\eta(\mathbf{x}).
\end{split}
\end{equation}
Hence $\ln\epsilon_0+\lambda h(\mathbf{x})\ge \eta(\mathbf{x})$ for all $\mathbf{x}\in\mathcal{C}$.
By Corollary~2, this implies $\varepsilon(\cdot)\in\mathcal{E}$, which yields
$\Pi_{\varepsilon}(\mathbf{x})\cap\mathcal{U}\neq\emptyset$ for all $\mathbf{x}\in\mathcal{C}$.
\end{proof}
Based on Proposition~\ref{paramter_sample_condition}, the inequalities~\eqref{sampling_Lip}
define a polyhedral feasible set in the parameter space $(\ln\epsilon_0,\lambda)$,
since each constraint is affine in $(\ln\epsilon_0,\lambda)$.
We compute an optimal parameter pair within this set by solving the following LP:
\begin{small}
\begin{equation}\label{eq:lp_param_design}
\begin{aligned}
\min_{\ln\epsilon_0,\lambda} & \,\,\ln\epsilon_0 + \rho\lambda\\
\mathrm{s.t.} 
& \,\ln\epsilon_0 + \lambda\big(h(\mathbf{x}_i)-L_h\kappa\big)
\ge \eta(\mathbf{x}_i) + L_\eta\kappa, i=1,\dots,N,\\
& \,\lambda \ge \lambda_{\min} > 0,
\end{aligned}
\end{equation}
\end{small}

\noindent
where $\lambda_{\min}>0$ is a prescribed lower bound ensuring the
monotonicity of the tuning function
$\epsilon_0 e^{\lambda h}$.
The weighting parameter $\rho\ge 0$ regulates the trade-off between
the baseline magnitude $\ln\epsilon_0$ and the growth rate $\lambda$
of the tuning function.
After solving~\eqref{eq:lp_param_design}, the optimal parameter pair
$(\ln\epsilon_0^\star,\lambda^\star)$ can be obtained.
By Proposition~\ref{paramter_sample_condition}, the resulting tuning
function satisfies $\varepsilon(\cdot)\in\mathcal{E}$ and therefore
ensures $\Pi_{\varepsilon}(\mathbf{x})\cap\mathcal{U}\neq\emptyset$
for all $\mathbf{x}\in\mathcal{C}$.

\begin{Rmk}\label{rmk:Lip_constants}
The conservativeness of the condition~\eqref{sampling_Lip} comes from the robustness margins induced by the covering radius $\kappa$ and the Lipschitz constants $(L_\eta,L_h)$.
Specifically, the bound $\eta(\mathbf{x})\le \eta(\mathbf{x}_i)+L_\eta\kappa$ accounts for the discretization
error of $\eta(\cdot)$, while the term $\lambda L_h\kappa$
robustifies the evaluation of $h(\cdot)$ when transferring the constraints from samples to the entire set $\mathcal{C}$.
If $\eta(\cdot)$ and $h(\cdot)$ are continuously differentiable on the compact set $\mathcal{C}$, valid Lipschitz
constants can be chosen as $L_\eta=\sup_{\mathbf{x}\in\mathcal{C}}\|\nabla \eta(\mathbf{x})\|$, $L_h=\sup_{\mathbf{x}\in\mathcal{C}}\|\nabla h(\mathbf{x})\|$, which are finite due to compactness. In practice, $L_\eta$ and $L_h$ can be estimated via analytical gradient bounds
or numerical evaluations. Moreover, since $\kappa$ decreases with the sampling density, the discretization gap
can be reduced by denser sampling or tighter Lipschitz estimates, leading to a less conservative parameter design.
\end{Rmk}

\subsection{Controller Synthesis}\label{subsec:controller_synthesis}
After obtaining the optimal parameters $(\ln\epsilon_0^*, \lambda^*)$ for the tuning function $\varepsilon(h(\mathbf{x}))$ by solving the LP \eqref{eq:lp_param_design}, we have established a rigorous framework that guarantees the existence of an admissible control input for all $\mathbf{x} \in \mathcal{C}$. To realize the control law $\mathbf{k}: \mathcal{C} \to \mathcal{U}$ established in Theorem~\ref{thm:input_constrained_tissf}, we utilize a safety-filter formulation based on a QP. Specifically, given a nominal performance-oriented feedback controller $\mathbf{k}_{\mathrm{nom}}: \mathbb{R}^n \to \mathbb{R}^m$, the actual control input $\mathbf{u}^\star$ is synthesized at each time step as follows:
\begin{subequations}\label{eq:TISSf_QP}
\begin{align}
\mathbf{u}^\star(\mathbf{x}) = \arg\min & \quad \frac{1}{2} \|\mathbf{u} - \mathbf{k}_{\mathrm{nom}}(\mathbf{x})\|^2 \label{eq:QP_obj} \\
\text{s.t.} & \quad c(\mathbf{x}) + \mathbf{d}(\mathbf{x}) \mathbf{u} \ge \frac{\|\mathbf{d}(\mathbf{x})\|_{2}^2}{\epsilon_0^* e^{\lambda^* h(\mathbf{x})}}, \label{eq:QP_TISSf_cons} \\
& \quad \mathbf{u} \in \mathcal{U}, \label{eq:QP_input_cons}
\end{align}
\end{subequations}
where $c(\mathbf{x})$ and $\mathbf{d}(\mathbf{x})$ are defined as in \eqref{HOCBF_ISSf}. Crucially, since the parameters of $\varepsilon(\cdot)$ are chosen such that $\varepsilon \in \mathcal{E}$, the intersection of the TISSf-CBF half-space and the input constraint set $\mathcal{U}$ is non-empty, i.e., $\Pi_{\varepsilon}(\mathbf x)\cap \mathcal U \neq \emptyset$, for all safe states $\mathbf{x}\in\mathcal{C}$. This ensures that the QP \eqref{eq:TISSf_QP} is recursively feasible by design.
\section{Simulation Results}\label{sec:sim_results}
In this section, we validate the proposed TISSf (LP-QP) framework using an adapted version of the CCC case study in~\cite{alan2021safe} and compare it against multiple baselines: (i) TISSf (Baseline)~\cite{alan2021safe} with manually tuned parameters; (ii) TISSf (Sat), which enforces the satisfaction of input constraints via saturation; and (iii) TISSf (Trial), which further tunes $(\epsilon_0,\lambda)$ via trial-and-error. Unlike these approaches, our design restricts the tuning function to the admissible set $\mathcal{E}$ characterized in Section~\ref{TISSf_with_Input_Limits}, ensuring a priori compatibility with the prescribed input constraints over the entire safety set $\mathcal{C}$.

Controller synthesis is performed on the longitudinal model:
\begin{equation}\label{eq:truck_model}
\dot D = v_L - v,\quad \dot v = u + \omega(t),\quad \dot v_L = a_L,
\end{equation}
where $D$, $v$, and $v_L$ denote headway distance, ego velocity, and lead velocity, respectively. The perturbation $\omega(t)=1.2$ satisfies $\|\omega\|_{\infty}\leq \delta=1.2$. The primary control objective is to match the velocity of the leading vehicle $v_L$ while maintaining a safe headway. This is characterized by the nominal controller $k_{\mathrm{nom}}(\mathbf{x}) = k_1(V_D(D) - v) + k_2(v_L - v)$, with $k_1 = 0.85$ and $k_2 = 0.75$. The desired velocity $V_D(D)$ is defined as $V_D(D) = \min\{ \max\{0, K_{v}(D - D_{\mathrm{st}})\}, \bar{v} \}$, where $D_{\mathrm{st}} = 7$~m is the standstill distance, $K_{v} = 0.7$ is the velocity gain, and $\bar{v} = 20$~m/s is the ego vehicle's speed limit. The simulation scenario involves a leading vehicle that initially travels at $v_L = 15$~m/s and initiates a braking maneuver at $t = 5$~s with $a_L = -4~\mathrm{m/s^2}$ until it stops (see Fig.~\ref{fig:simulation} (b)). The ego vehicle's acceleration $u$ is strictly constrained to $[\underline{a}, \bar{a}]$ with $\underline{a} = -6~\mathrm{m/s^2}$ and $\bar{a} = 0.8~\mathrm{m/s^2}$. Safety is encoded by $h(\mathbf{x}) = D - \hat h(v,v_L) \ge 0$, where the safe distance policy $\hat{h}(v, v_L) = D_{sf} + \theta v + \eta v_L + \xi v^2 + \zeta v v_L + \omega v_L^2$ follows the parameters: $D_{\mathrm{sf}} = 2$, $\theta = 1.1$, $\eta = 0.6$, $\xi = 0.03$, $\zeta = -0.03$, and $\omega = -0.03$. 

To implement the proposed safety-critical control, we first solve the offline robust sampled LP \eqref{eq:lp_param_design} to obtain the optimal tuning parameters $(\ln\epsilon_0^*, \lambda^*)$. Following Proposition~\ref{paramter_sample_condition}, the safe set $\mathcal{C}$ is discretized. We restrict $\mathcal{C}\subset\mathbb{R}^3$ to the compact domain 
$D\in[0,60]$~m, $v\in[0,\bar v]$~m/s, and $v_L\in[0,\bar v]$~m/s,
and sample $N=2000$ Latin-hypercube points with $\kappa=0.1$, where $L_h=\sup_{\mathbf{x}\in\mathcal{C}}\|\nabla h(\mathbf{x})\|$ and $L_\eta=\sup_{\mathbf{x}\in\mathcal{C}}\|\nabla \eta(\mathbf{x})\|$ are approximated by the maximum gradient norms over the samples. We use $L_h = 0.2$ and $L_\eta = 0.3$ to robustify the compatibility condition against discretization errors. With a weighting parameter $\rho = 12$, the LP yields $\epsilon_0^* = 5.6 \times 10^{-3}$ and $\lambda^* = 0.18$. These parameters define $\varepsilon(h) = \epsilon_0^* e^{\lambda^* h}$ and guarantee $\Pi_{\varepsilon}(\mathbf{x}) \cap \mathcal{U} \neq \emptyset$. The real-time control input is then synthesized by the QP-based safety filter \eqref{eq:TISSf_QP}. For comparison, TISSf (Baseline)~\cite{alan2021safe} uses $(\epsilon_0,\lambda)=(5\times10^{-4},0.1)$ and the fixed-form control law
$u = k_{\text{nom}}(\mathbf{x}) -1/(\epsilon_0 \mathrm{e}^{\lambda\left(D-\hat{h}\left(v, v_{L}\right)\right)}) \frac{\partial \hat{h}}{\partial v}\left(v, v_{L}\right)$, while TISSf (Sat) applies saturation to enforce $u\in[\underline a,\bar a]$, and TISSf (Trial) selects $(\epsilon_0,\lambda)$ by trial-and-error and stops when both $h(\mathbf{x})+\zeta(h(\mathbf{x}),\delta)>0$ and $u\in[\underline a,\bar a]$ are satisfied.

The closed-loop performance is illustrated in Fig.~\ref{fig:simulation}. In Fig.~\ref{fig:simulation}(a), Proposed (LP-QP) maintains the smallest headway while preserving a safe distance, namely $D\ge 0$. TISSf (Trial) yields a larger headway than Proposed (LP-QP) but remains less conservative than TISSf (Baseline) and TISSf (Sat), which keep substantially larger distances. In Fig.~\ref{fig:simulation}(b), Proposed (LP-QP) and TISSf (Trial) achieve fast velocity tracking, whereas TISSf (Sat) exhibits noticeably slower transients under the tightened input limit. TISSf (Baseline) shows the most aggressive velocity response and may exhibit overshoot during the maneuver. Fig.~\ref{fig:simulation}(c) shows that the robust safety metric $h(\mathbf{x})+\zeta(h(\mathbf{x}),\delta)$ remains positive for all approaches, while TISSf (Baseline) is closer to the boundary. In Fig.~\ref{fig:simulation}(d), Proposed (LP-QP), TISSf (Sat), and TISSf (Trial) satisfy $u\in[\underline a,\bar a]$ throughout the maneuver, whereas TISSf (Baseline) can exhibit transients beyond the bounds. Overall, Proposed (LP-QP) achieves the best headway efficiency while enforcing safety and input constraints by design.
\vspace{-6pt}
\section{Conclusion}
This paper addresses a critical gap in the TISSf framework, where robust safety guarantees often conflict with input constraints. By reformulating the compatibility requirement as a tuning function design objective, we utilize support functions to derive an exact lower bound on the tuning function. This characterization provides a rigorous constructive procedure for the synthesis of tuning functions, ensuring input compatibility and recursive feasibility of the resulting QP-based safety filter under input constraints. Moreover, the proposed framework offers a unified treatment for different input constraints, as demonstrated through specialized conditions for norm-bounded, polyhedral, and box sets. Furthermore, by leveraging a covering-based sampling strategy, we develop a tractable offline LP procedure that upgrades pointwise compatibility to a formal compatibility guarantee over the entire safe set. Simulation results for a CCC application demonstrate that our approach resolves the compatibility issues in unconstrained TISSf synthesis by design.
\begin{figure}[tp]
    \centering
    \includegraphics[width=0.48\textwidth]{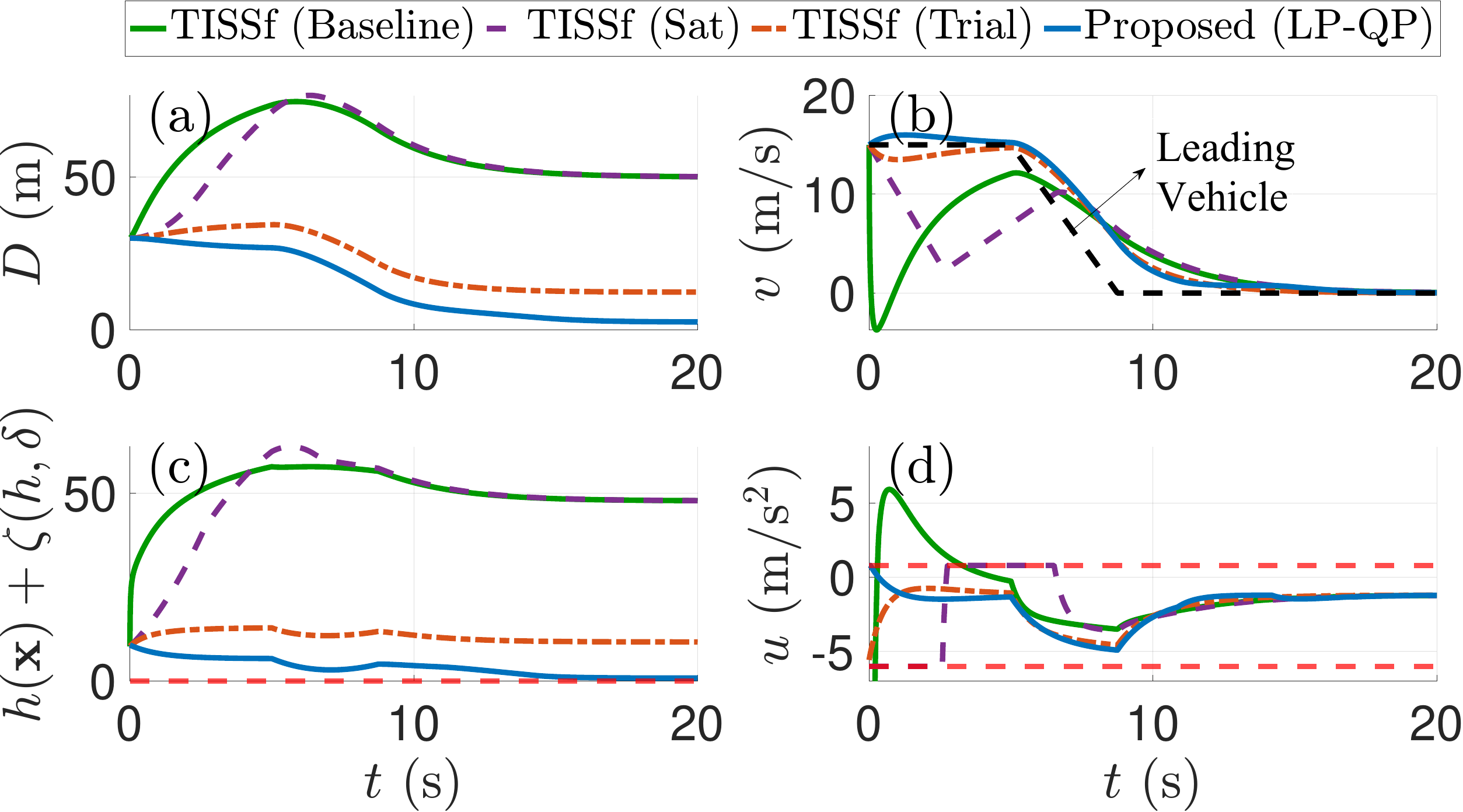}
  \caption{Comparative CCC simulation results under perturbations: (a) headway distance $D$; (b) ego vehicle velocity $v$ (the leading-vehicle profile is shown in the plot); (c) robust safety metric $h(\mathbf{x})+\zeta(h(\mathbf{x}),\delta)$, where values above zero indicate robust safety; and (d) control input $u(t)$ with acceleration limits (red dashed).}
    \label{fig:simulation}
\end{figure}



\bibliographystyle{IEEEtran}
\bibliography{Tunable_ISSf_Input.bib}

@article{Zeroing_CBF,
  title={Control barrier function based quadratic programs for safety critical systems},
  author={Ames, Aaron D and Xu, Xiangru and Grizzle, Jessy W and Tabuada, Paulo},
  journal={IEEE Transactions on Automatic Control},
  volume={62},
  number={8},
  pages={3861--3876},
  year={2016},
  publisher={IEEE}
}

@article{alan2021safe,
  title={Safe controller synthesis with tunable input-to-state safe control barrier functions},
  author={Alan, Anil and Taylor, Andrew J and He, Chaozhe R and Orosz, G{\'a}bor and Ames, Aaron D},
  journal={IEEE Control Systems Letters},
  volume={6},
  pages={908--913},
  year={2021},
  publisher={IEEE}
}

@inproceedings{agrawal2021safe,
  title={Safe control synthesis via input constrained control barrier functions},
  author={Agrawal, Devansh R and Panagou, Dimitra},
  booktitle={2021 60th IEEE Conference on Decision and Control (CDC)},
  pages={6113--6118},
  year={2021},
  organization={IEEE}
}

@article{breeden2023robust,
  title={Robust control barrier functions under high relative degree and input constraints for satellite trajectories},
  author={Breeden, Joseph and Panagou, Dimitra},
  journal={Automatica},
  volume={155},
  pages={111109},
  year={2023},
  publisher={Elsevier}
}

@article{mestres2022optimization,
  title={Optimization-based safe stabilizing feedback with guaranteed region of attraction},
  author={Mestres, Pol and Cort{\'e}s, Jorge},
  journal={IEEE Control Systems Letters},
  volume={7},
  pages={367--372},
  year={2022},
  publisher={IEEE}
}

@book{vershynin2018high,
  title={High-Dimensional Probability: An Introduction with Applications in Data Science},
  author={Vershynin, Roman},
  year={2018},
  publisher={Cambridge University Press}
}

@article{kolathaya2018input,
  title={Input-to-state safety with control barrier functions},
  author={Kolathaya, Shishir and Ames, Aaron D},
  journal={IEEE Control Systems Letters},
  volume={3},
  number={1},
  pages={108--113},
  year={2018},
  publisher={IEEE}
}

@article{xiao2021adaptive,
  title={Adaptive control barrier functions},
  author={Xiao, Wei and Belta, Calin and Cassandras, Christos G},
  journal={IEEE Transactions on Automatic Control},
  volume={67},
  number={5},
  pages={2267--2281},
  year={2021},
  publisher={IEEE}
}

@article{wieland2007constructive,
  title={Constructive safety using control barrier functions},
  author={Wieland, Peter and Allg{\"o}wer, Frank},
  journal={IFAC Proceedings Volumes},
  volume={40},
  number={12},
  pages={462--467},
  year={2007},
  publisher={Elsevier}
}

@article{tan2021high,
  title={High-order barrier functions: Robustness, safety, and performance-critical control},
  author={Tan, Xiao and Cortez, Wenceslao Shaw and Dimarogonas, Dimos V},
  journal={IEEE Transactions on Automatic Control},
  volume={67},
  number={6},
  pages={3021--3028},
  year={2021},
  publisher={IEEE}
}

@article{xiao2021high,
  title={High-order control barrier functions},
  author={Xiao, Wei and Belta, Calin},
  journal={IEEE Transactions on Automatic Control},
  volume={67},
  number={7},
  pages={3655--3662},
  year={2021},
  publisher={IEEE}
}

@article{lopez2020robust,
  title={Robust adaptive control barrier functions: An adaptive and data-driven approach to safety},
  author={Lopez, Brett T and Slotine, Jean-Jacques E and How, Jonathan P},
  journal={IEEE Control Systems Letters},
  volume={5},
  number={3},
  pages={1031--1036},
  year={2020},
  publisher={IEEE}
}

@inproceedings{taylor2020learning,
  title={Learning for safety-critical control with control barrier functions},
  author={Taylor, Andrew and Singletary, Andrew and Yue, Yisong and Ames, Aaron},
  booktitle={Learning for dynamics and control},
  pages={708--717},
  year={2020},
  organization={PMLR}
}

@inproceedings{castaneda2021gaussian,
  title={Gaussian process-based min-norm stabilizing controller for control-affine systems with uncertain input effects and dynamics},
  author={Castaneda, Fernando and Choi, Jason J and Zhang, Bike and Tomlin, Claire J and Sreenath, Koushil},
  booktitle={2021 American Control Conference (ACC)},
  pages={3683--3690},
  year={2021},
  organization={IEEE}
}

@inproceedings{romdlony2016new,
  title={On the new notion of input-to-state safety},
  author={Romdlony, Muhammad Zakiyullah and Jayawardhana, Bayu},
  booktitle={2016 IEEE 55th Conference on Decision and Control (CDC)},
  pages={6403--6409},
  year={2016},
  organization={IEEE}
}

@article{alan2023control,
  title={Control barrier functions and input-to-state safety with application to automated vehicles},
  author={Alan, Anil and Taylor, Andrew J and He, Chaozhe R and Ames, Aaron D and Orosz, G{\'a}bor},
  journal={IEEE Transactions on Control Systems Technology},
  volume={31},
  number={6},
  pages={2744--2759},
  year={2023},
  publisher={IEEE}
}

@article{jankovic2018robust,
  title={Robust control barrier functions for constrained stabilization of nonlinear systems},
  author={Jankovic, Mrdjan},
  journal={Automatica},
  volume={96},
  pages={359--367},
  year={2018},
  publisher={Elsevier}
}

\end{document}